\documentclass[letterpaper, 10 pt, conference]{ieeeconf} 
\usepackage{silence}
\IEEEoverridecommandlockouts
\usepackage[noadjust]{cite}
\usepackage{amsmath,amssymb,amsfonts}
\usepackage{amsthm}
\usepackage{algorithmic}
\usepackage{algorithm}
\usepackage{graphicx}
\usepackage{graphbox}
\usepackage{textcomp}
\usepackage{mathtools}
\usepackage{multirow}
\usepackage{balance}
\usepackage[export]{adjustbox}
\usepackage{color}
\usepackage{subcaption}
\usepackage{accents}
\usepackage{textcomp}
\usepackage{comment}
\usepackage{bm}
\usepackage{upgreek}
\let\labelindent\relax
\usepackage[shortlabels]{enumitem}
\usepackage[dvipsnames]{xcolor}

\definecolor{LightCyan}{rgb}{0.85,0.85,1.0}
\newcommand{\proj}{\text{ \normalfont{proj}}}
\renewcommand{\det}{\text{ \normalfont{det}}}

\newcommand{\tr}{\normalfont\text{\normalfont tr}}
\newcommand{\mfs}[1]{{\normalfont\textsf{#1}}}
\newcommand{\mf}{\mathbf}

\DeclareMathOperator*{\argmin}{arg\,min}

\newtheorem{theorem}{Theorem}
\newtheorem{lemma}{Lemma}
\newtheorem{remark}{Remark}
\newtheorem{definition}{Definition}
\newtheorem{proposition}{Proposition}
\newtheorem{assumption}{Assumption}

\def\RevOne#1{{#1}}
\def\RevTwo#1{{#1}}
\def\RevThree#1{{#1}}
\def\RevAll#1{{#1}}

\makeatletter
\let\NAT@parse\undefined
\makeatother
\usepackage[hidelinks]{hyperref}

\title{\LARGE \bf
\RevThree{Distributed outer approximation of the intersection of ellipsoids}}

\author{Rodrigo Aldana-L\'{o}pez,  Eduardo Sebasti\'{a}n, Rosario Aragü\'{e}s, Eduardo Montijano and Carlos Sagü\'{e}s
\thanks{\textcolor{red}{This is the accepted version of the manuscript: ``Distributed outer approximation of the intersection of ellipsoids," Rodrigo Aldana-L\'{o}pez,  Eduardo Sebasti\'{a}n, Rosario Aragü\'{e}s, Eduardo Montijano and Carlos Sagü\'{e}s, in IEEE Control Systems Letters, 2023, DOI: 10.1109/LCSYS.2023.3280259. 
\textbf{Please cite the publisher's version}. For the publisher's version and full citation details see:
\url{https://doi.org/10.1109/LCSYS.2023.3280259}. 
}
}
\thanks{This work was supported by ONR Global
grant N62909-19-1-2027, the Spanish projects PID2021-125514NB-I00, PID2021-124137OB-I00, TED2021-130224B-I00
funded by MCIN/AEI/10.13039/501100011033, by ERDF A way of making Europe and by the
European Union NextGenerationEU/PRTR, Project DGA T45-20R by the Gobierno de Aragón,  by the Universidad de Zaragoza and Banco Santander, by CONACYT-Mexico grant 739841 and Spanish grant FPU19-05700.}%
\thanks{All the authors are with the Instituto de Investigaci\'on en Ingenier\'ia de Arag\'on, Universidad de Zaragoza, Spain 
(email:\texttt{\scriptsize rodrigo.aldana.lopez@gmail.com, esebastian@unizar.es, raragues@unizar.es, emonti@unizar.es, csagues@unizar.es})}%
\thanks{``© 2023 IEEE.  Personal use of this material is permitted.  Permission from IEEE must be obtained for all other uses, in any current or future media, including reprinting/republishing this material for advertising or promotional purposes, creating new collective works, for resale or redistribution to servers or lists, or reuse of any copyrighted component of this work in other works.”}
}

\begin{document}
\maketitle


\begin{abstract}
\RevOne{The outer Löwner-John method is widely used in sensor fusion applications to find the smallest ellipsoid that can approximate the intersection of a set of ellipsoids, described by positive definite covariance matrices modeling the quality of each sensor. We propose a distributed algorithm to solve this problem when these matrices are defined over the network's nodes. This is of particular significance as it is the first decentralized algorithm capable of computing the covariance intersection ellipsoid by combining information from the entire network using only local interactions.} The solution is based on a reformulation of the centralized problem, leading to a local protocol based on exact dynamic consensus tools. After reaching consensus, the protocol converges to an outer Löwner-John ellipsoid in finite time, and to the global optimum asymptotically. Formal convergence analysis
and numerical experiments are provided to validate the
proposal’s advantages.
\end{abstract}

\section{Introduction}\label{sec:intro}
The Löwner-John (L-J) methods~\cite{Henrion2001LMI} are a series of ellipsoidal approximations for convex sets. Of particular interest are the convex sets generated by the intersection of $n$-dimensional ellipsoids, described by symmetric and positive semidefinite $n$-dimensional matrices. L-J methods have a significant presence in a wide variety of applications~\cite{Lasserre2015Generalization}, such as robust control~\cite{Henrion2011Ellipsoidal} or statistical analysis~\cite{ Lu2018Relatively}. Specially important is in the field of data fusion and state estimation~\cite{Julier2017General}, where the ellipsoids represent the measurements or estimates' uncertainty. \RevOne{However, there are no distributed algorithms to compute L-J ellipsoids, despite their potential application in sensor networks, where the individual measurements are scattered over a communication network \cite{Sebastian2021CDC}.} 

In sensor fusion, both the inner and outer L-J ellipsoidal approximations are widely explored under the name of Covariance Intersection Method (CIM)~\cite{Julier2017General}. The CIM is posed as a convex optimization program yielding the largest outer approximation of the intersection of two ellipsoids described by two covariance matrices. Since then, different works have applied variants of the CIM. For example, a sequential procedure of fusion of two ellipsoids~\cite{Deng2012sequential}  can be used to extend the CIM to $\mfs{N}$ ellipsoids \cite{Niehsen2002}. On the other hand, recent distributed Kalman filters~\cite{Hu2011Diffusion, Wang2017convergence, He2018Consistent} propose CIM variants which directly consider $\mfs{N}$ ellipsoids via approximations. More recently, the outer L-J method~\cite{Sebastian2021CDC} has been explored for the fusion of the estimates and to certify the estimation with optimality guarantees. All these works are particular instances of L-J methods. However, rather than finding the global L-J ellipsoids in a distributed manner, they use the CIM to approximate the fusion of neighboring estimates.

The computation of L-J ellipsoids can be posed as an optimization program \cite{Henrion2001LMI}. In this sense, distributed optimization methods~\cite{Nedic2018Distributed} have been intensively researched in recent years, \RevTwo{where a general assumption~\cite{Boyd2011Distributed, Ning2017Distributed,Ma2017Distributed,Esteki2022Distributed} is that the global cost function is the sum of local objectives.} In the case of L-J methods, the cost function is not separable. This separability issue extends to the constraints of the outer L-J methods, where a coupled equality constraint holds. \RevTwo{To address coupled equality constraints in sensor fusion, the Inverse CIM~\cite{noack2017decentralized} computes local bounds to ensure consistency, whereas other methods rely on fusion centers~\cite{Chen2021Distributed} to gather the distributively pre-processed covariance matrices. Instead, we propose a distributed projected gradient flow \cite{projGrad} protocol that converges to the coupled equality constraint in a prescribed time. In more general settings, 
} the literature proposes distributed dual sub-gradient algorithms~\cite{romao2021,camisa2022}, \RevTwo{primal relaxations~\cite{notarnicola2019constraint},} the so-called ``subgradient push''~\cite{Nedic2014Distributed} or neural-network-based approaches~\cite{Le2017Neurodynamic}. \RevTwo{In contrast, our method is suitable for non-separable optimization objectives.}

Motivated by this discussion, we develop, for the first time, a distributed protocol that computes the outer approximation of the intersection of ellipsoids when only local information and distributed communications are available. The proposal is based on a continuous-time Exact Dynamic Consensus protocol (EDC), which extends previous protocols \cite{edcho,freeman2019, Aldana2019} to converge before a prescribed time. Consequently, we can certify the moment in which an outer L-J ellipsoid is already available. The quality of the outer L-J ellipsoid is further improved through a distributed optimization step based on the projected gradient flow \cite{projGrad}, such that the global optimum is asymptotically found. We discuss how our proposal can be used for CIM in sensor fusion. 


\textbf{Notation: } $\tr(\bullet), \det(\bullet)$ denote the trace and determinant. $\mathbb{R},\mathbb{R}_{\geq 0}, \mathbb{R}_{>0}$ denote the real, non-negative and positive reals respectively.  $0\preceq \mf{P}$ denotes when a matrix $\mf{P}\in\mathbb{R}^{n\times n}$ is positive semi-definite. $\mf{0}, \mf{I}$ denote the zero and identity matrices. \RevOne{For an undirected graph $\mathcal{G}$, denote with $\lambda_{\mathcal{G}}$ its standard algebraic connectivity (see \cite{Aldana2019})}.

\section{Problem statement}
\begin{table}
\vspace{1em}
\centering
\renewcommand{\arraystretch}{1.5}
\begin{tabular}{|| c| c | c||} \hline
  $\mu$ & $f(\mf{Q}(\mf{x}))$ & $g_i(x_i,\mf{Q}(\mf{x}))$ \\\hline
   0 & $-\tr(\mf{Q}(\mf{x}))$ & $-2x_i\tr(\mf{P}_i^{-1})$\\
 1 & $\log(\!\!\det(\mf{Q}(\mf{x})^{-1}))$ & $-2x_i\tr(\mf{Q}(\mf{x})^{-1}\mf{P}_i^{-1})$ \\
  2 &$\tr(\mf{Q}(\mf{x})^{-1})$ & $-2x_i\tr(\mf{Q}(\mf{x})^{-1}\mf{P}_i^{-1}\mf{Q}(\mf{x})^{-1})$ \\\hline
\end{tabular}
\caption{Different \RevThree{ellipsoidal size measure}  functions $f(\bullet)$ with $\mf{Q}(\mf{x}) =(1/\mfs{N})\sum_{i=1}^\mfs{N}x_i^2\mf{P}_i^{-1}$ as well as the components of the gradient $\nabla f(\mf{Q}(\mf{x})) = [g_1(x_1,\mf{Q}(\mf{x})),\dots,g_\mfs{N}(x_\mfs{N},\mf{Q}(\mf{x}))]^\top/\mfs{N}$. We label each choice by $\mu$ for its reference in Theorem \ref{th:main}.}
\label{tab:options}
\renewcommand{\arraystretch}{1}
\end{table}

\RevOne{Consider a set of $\mfs{N}$ agents which communicate through a communication network modeled by an undirected connected graph $\mathcal{G}=(\mathcal{I}, \mathcal{F})$, where $\mathcal{I}$ and $ \mathcal{F}$ are the sets of nodes and edges respectively. We denote by $\mathcal{N}_i\subseteq\mathcal{I}$ the index set of neighbors for agent $i\in\mathcal{I}$. Each agent $i$ is described by a $n-$dimensional ellipsoidal set $\mathcal{E}(\mf{P}_i)=\{\mf{y}\in\mathbb{R}^n : \mf{y}^\top\mf{P}^{-1}_i\mf{y}\leq 1\}$, given $\mf{0}\prec\mf{P}_i\in\mathbb{R}^{n\times n}$. }

The goal for the agents is to cooperate to find an ellipsoidal set $\mathcal{E}(\mf{P})$ covering $\check{\mathcal{E}}:=\bigcap_{i=1}^\mfs{N}\mathcal{E}(\mf{P}_i)$, characterized by some $\mf{P}\succ 0$. A wide family of ellipsoidal sets covering $\check{\mathcal{E}}$ is the one parameterized by $\mf{P}(\bm{\lambda})^{-1} := \sum_{i=1}^\mfs{N}\lambda_i\mf{P}_i^{-1}$ with $\sum_{i=1}^\mfs{N}\lambda_i = 1$ and $\bm{\lambda}=[\lambda_1,\dots,\lambda_\mfs{N}]^\top\in\mathbb{R}^\mfs{N}_{\geq 0}$. It can be verified that $\mathcal{E}(\mf{P}(\bm{\lambda}))\supset \check{\mathcal{E}}$ ~\cite{Henrion2001LMI}. Therefore, designing $\mf{P}(\bm{\lambda})^{-1}$ as a convex combination of $\{\mf{P}_i^{-1}\}_{i=1}^\mfs{N}$ defines a family of L-J outer ellipsoids for the intersection of $\{\mathcal{E}(\mf{P}_i)\}_{i=1}^\mfs{N}$. The weights $\bm{\lambda}$ can then be optimized by solving the following optimization program:
\begin{equation}
\label{eq:original}
\begin{aligned}
&\min_{\bm{\lambda}\in\mathbb{R}^\mfs{N}_{\geq 0}} f(\mf{P}(\bm{\lambda})^{-1}),\ \  
\text{such that\ \ }  \sum_{i=1}^\mfs{N} \lambda_i = 1,
\end{aligned}
\end{equation}
where $f(\bullet)$ is a function that measures the size of $\mathcal{E}(\mf{P}(\bm{\lambda}))$. For instance, $f(\bullet) = \log(\!\!\det(\bullet))$ can be used to minimize the volume of $\mathcal{E}(\mf{P}(\bm{\lambda}))$. Other popular choices of $f(\bullet)$ are in Table \ref{tab:options}. The purpose of this work is to design a distributed protocol that finds the optimum of \eqref{eq:original}. \RevOne{Due to the numerical difficulties found when dealing with equality constraints in practice, we recast \eqref{eq:original} into}
\begin{equation}
\label{eq:relax}
\begin{aligned}
&\min_{\mf{x}\in\mathcal{C}} f(\mf{Q}(\mf{x})), \ \  \mf{Q}(\mf{x}):= \frac{1}{\mfs{N}}\sum_{i=1}^\mfs{N} x_i^2\mf{P}^{-1}_i,
\end{aligned}
\end{equation}
\RevOne{where instead of an equality constraint, the solution is restricted to a wider feasible manifold $\mathcal{C} = \{\mf{x}\in\mathbb{R}^\mfs{N} : 1-\varepsilon \leq s(\mf{x})\leq 1\}$ of $\mf{x}=[x_1,\dots,x_\mfs{N}]^{\top}$, $s(\mf{x}):= \|\mf{x}\|^2/\mfs{N}$ and $\varepsilon \in (0,1)$}. By using the change of coordinates $\lambda_i = x_i^2/\mfs{N}$, \RevOne{we ensure $\lambda_i\in\mathbb{R}_{\geq 0}$} and that the unique minimizer of \eqref{eq:original} follows from the minimizer of \eqref{eq:relax} when $\varepsilon=0$. For any other $\varepsilon>0$, we use \eqref{eq:relax} to approximate solutions of \eqref{eq:original} with arbitrary accuracy dictated by the size of $\varepsilon$. In the following section, we provide a distributed algorithm to solve \eqref{eq:relax}.

\section{Distributed outer ellipse computation}

To solve problem~\eqref{eq:relax} in a distributed fashion, we propose a novel distributed protocol based on Projected Gradient Flow (PGF)~\cite{projGrad} methods. The idea is to drive the trajectories of $\mf{x}(t)$ towards the feasible manifold where the equality constraint in \eqref{eq:original} holds. Once there, the trajectories of $\mf{x}(t)$ flow towards the optimum while fulfilling the equality constraint. To do so, we set a suitable virtual control action $\dot{x}_i(t)=u_i(t)$ using only local information.

To design a PGF-based protocol, we first need agreement on some global quantities across the nodes of the network. 
Therefore, the first stage of our algorithm computes local estimates $\{ \hat{s}_i(t),\hat{\mf{Q}}_i(t)\}_{i=1}^\mfs{N}$ for $s(\mf{x}(t)),\mf{Q}(\mf{x}(t))$ using the following EDC protocols:
\begin{equation}
\label{eq:consensus1}
\begin{aligned}
\hat{s}_i(t) &= x_i(t)^2 - {v}_i(t)\\
\dot{v}_i(t) &=
\kappa_s\sum_{\RevTwo{j}\in\mathcal{N}_i} 
\phi\left(\hat{s}_j(t)-\hat{s}_i(t) ;\zeta_s ; q \right)\\
\hat{\mf{Q}}_i(t) &= x_i(t)^2\mf{P}_i^{-1} - {\mf{V}}_i(t)\\
\dot{{\mf{V}}}_i(t) &= \kappa_{\mf{Q}}\sum_{\RevTwo{j}\in\mathcal{N}_i} \phi\left(\hat{\mf{Q}}_j(t)-\hat{\mf{Q}}_i(t);\zeta_\mf{Q} ; q \right) 
\end{aligned}
\end{equation}
\RevOne{with auxiliary variables $\mf{V}_i(t),v_i(t)$ initialized as ${\mf{V}}_i(0)=\mf{0}, {v}_i(0)=0$}. Moreover, $\kappa_{\mf{Q}}, \kappa_s,\zeta_\mf{Q},\zeta_s>0, q\in(0,1)$ are design parameters. We use $\phi(\bullet;\zeta;q) = (|\bullet|^{1-q} + |\bullet|^{1+q} + \zeta)\text{\normalfont sign}(\bullet)$ for a scalar parameter $\bullet$, and component-wise for $\bullet$ of any other dimension. As will be proven in Section \ref{sec:proof}, after a transient of prescribed duration $T_c$, these estimates will comply $\hat{\mf{Q}}_i(t)\equiv \mf{Q}(\mf{x}(t)), \hat{s}_i(t) \equiv s(\mf{x}(t)), \forall t\geq T_c$ for suitable $\kappa_{\mf{Q}}, \kappa_s$. \RevTwo{This is possible since $\zeta_\mf{Q}, \zeta_s$ introduce a discontinuous sliding mode term in $\phi$ allowing \eqref{eq:consensus1} to achieve EDC even with time-varying consensus inputs}. Then, all agents estimate its local component of the gradient $g_i(x_i(t),{\mf{Q}}(\mf{x}(t)))\equiv g_i(x_i(t),\hat{\mf{Q}}_i({t})), \forall t\geq T_c$, as in Table~\ref{tab:options}. 

During the first consensus stage defined in \eqref{eq:consensus1}, we set a control action ${u}_i(t)=0, \forall t\in[0,T_c]$. The second stage of our algorithm consists of taking the arbitrary initial conditions $x_i(0)=x_i(T_c)$ and update $x_i(t)$ towards $\mathcal{C}$. Then, nodes do PGF to find the global optimum of~\eqref{eq:relax}. This is achieved by a discontinuous controller. For all $t\geq T_c$:
\begin{equation}
\label{eq:controller}
\kern -0.4cm {u}_i(t) \kern -0.1cm = \kern -0.1cm 
\begin{cases}
\kappa_\mathcal{C}x_i(t)\text{\normalfont sign}\kern -0.1cm \left(\left(1\kern -0.1cm -\kern -0.1cm \frac{\varepsilon}{2}\right)\kern -0.1cm -\kern -0.1cm \hat{s}_i(t)\right), &  \kern -0.3cm  \hat{s}_i(t)\RevThree{\notin}[1-\varepsilon,1]  \\
-\kappa_\mathcal{C}g_i(x_i(t),\hat{\mf{Q}}_i(t)), & \kern -0.3cm \hat{s}_i(t)\RevThree{\in}[1-\varepsilon,1] 
\end{cases}
\end{equation}
\RevThree{with design parameter $\kappa_\mathcal{C}>0$}, which leads to:
\begin{equation}
\label{eq:ideal}
\kern -0.2cm \dot{\mf{x}}(t) = 
\begin{cases}
\kappa_\mathcal{C}\mf{x}(t)\text{\normalfont sign}\left(\left(1-\frac{\varepsilon}{2}\right)-{s}(\mf{x}(t))\right), & \mf{x}(t)\notin\mathcal{C}  \\
-\kappa_\mathcal{C}\mfs{N}\nabla f(\mf{Q}(\mf{x}(t))), & \mf{x}(t)\in\mathcal{C}
\end{cases}
\end{equation}
under the synchronization conditions  $\hat{\mf{Q}}_i(t)\equiv \mf{Q}(\mf{x}(t))$, $\hat{s}_i(t) \equiv s(\mf{x}(t))$.
\RevOne{Using the idea of PGF, when $\mf{x}(t)\notin\mathcal{C}$, $\mf{x}(t)$ flows towards $\mathcal{C}$, in order to fulfill the equality constraint of the outer L-J method in \eqref{eq:original}. On the other hand, for $\mf{x}(t)\in\mathcal{C}$, $\mf{x}(t)$ flows in the direction opposite to the gradient, towards the optimum of the outer L-J method in \eqref{eq:original}. As we discuss in Section \ref{sec:proof}, this results in $\mf{x}(t)$ flowing in the direction of the projected gradient of $f(\mf{Q}(\mf{x}(t)))$ with respect to $\mathcal{C}$ in any case, allowing to maintain feasible trajectories and converge to the optimum.} In the following, we state our main result as well as a practical assumption under which it holds.

\begin{assumption} \label{ass:boundness}
Let $0<\underline{b}<\overline{b}, \overline{b}>1$ and $\underline{\sigma}, \overline{\sigma}>0$. Then, $x_i(0)\in[\underline{b},\overline{b}]$ and $\underline{\sigma}\mf{I}\preceq \mf{P}_i^{-1}\preceq\overline{\sigma}\mf{I}, \forall i\in\mathcal{I}$, and with $p$ the maximum scalar component among all $\{\mf{P}_i\}_{i=1}^\mfs{N}$.  
\end{assumption}


\begin{theorem}
\label{th:main}
Let $\mathcal{G}$ be a connected undirected graph with $\mfs{N}$ nodes, $\ell$ edges, algebraic connectivity $\lambda_{\mathcal{G}}$, and consider protocols \eqref{eq:consensus1}. Let $\kappa_\mathcal{C}>0, \varepsilon\in(0,1)$, and Assumption~\ref{ass:boundness} hold. Let  $f(\bullet)$ in Table \ref{tab:options} be labeled by $\mu\in\{0,1,2\}$ and 
$$
h(\mfs{N}) = \kappa_\mathcal{C} \max\{\sqrt{\mfs{N}}\overline{b}, 2\overline{b}\overline{\sigma}\mfs{N}^{\mu+1}(\underline{\sigma}\min\{\underline{b}^2,1-\varepsilon\})^{-\mu}\}.
$$
Let $\dot{x}_i(t)=u_i(t)$ with $u_i(t)=0, \forall t\in[0,T_c]$ and $u_i(t)$ defined as in \eqref{eq:controller} for $t\geq T_c$. For any $\kappa_{\mathcal{C}}>0$, if
\begin{equation}
\kappa_s,\kappa_{\mf{Q}}>\frac{\ell\pi}{q\lambda_\mathcal{G}T_c},\  \zeta_{s}>\frac{4\overline{b}h(\mfs{N})}{\kappa_s\sqrt{\lambda_\mathcal{G}}} ,\ \zeta_{\mf{Q}}>\frac{4p\overline{b}h(\mfs{N})}{\kappa_{\mf{Q}}\sqrt{\lambda_\mathcal{G}}},
\end{equation}
then, there exists $T_\varepsilon>0$ such that $\mf{x}(t)\in\mathcal{C}$, $\forall t\geq T_c+T_\varepsilon$. In addition, $\mathcal{E}(\mf{Q}(\mf{x}(t))^{-1})\supset \check{\mathcal{E}}, \forall t\geq T_c+T_\varepsilon$ and $f(\mf{Q}(\mf{x}(t)))$ converges asymptotically towards the optimum of \eqref{eq:relax}.
 \end{theorem}



\begin{remark}
\label{rem:accuracy}
\RevTwo{After consensus has been reached for $t\geq T_c$, each agent can check the condition $\hat{s}_i(t)=s(\mf{x}(t))\in[1-\varepsilon,1]$ to verify if $\mf{x}(t)\in\mathcal{C}$ and compute $T_\varepsilon$.} Hence, the proposed algorithm obtains an outer L-J ellipsoid from $t=T_c + T_\varepsilon$, since $\mf{x}(t)\in\mathcal{C}$ $\forall t \geq T_c + T_\varepsilon$. Thus, the ellipsoid $\mathcal{E}(\hat{\mf{Q}}_i(t)^{-1}) = \mathcal{E}(\mf{Q}(\mf{x}(t))^{-1})$ is already valid to use. Note that the greater $t$, the tighter $\mathcal{E}(\mf{Q}(\mf{x}(t))^{-1})$, so the user can design the algorithm depending on the computing resources available and the desired accuracy for the L-J ellipsoid.
\end{remark}

\begin{remark}
\label{rem:parameters}
\RevAll{The powers $1-q, 1+q$ with $q\in(0,1)$ in $\phi$ are known to induce a fixed settling time bound, regardless of the initial conditions. This allows all agents to make sure consensus has been reached before $t=T_c$ without having to individually check this condition, avoiding additional synchronization issues. The function $\phi$ along with $q$ and $\kappa_s, \zeta_s$ (similarly $\kappa_\mf{Q}, \zeta_\mf{Q}$) were introduced for static consensus with disturbances in \cite{Aldana2019}.}
\end{remark}

\subsection{\RevOne{Application to sensor fusion}}
\label{sec:fusion}

\RevOne{In this section, we describe how our proposal is applied to sensor fusion, and how it relates to other methods in this context. A standard sensor fusion problem consists of estimating a quantity of interest $\mf{p}\in\mathbb{R}^n$ where each agent has access to an estimate $\hat{\mf{p}}_i\in\mathbb{R}^n$ of $\mf{p}$ under Gaussian uncertainty with covariance $\mf{P}_i$. For example, such estimates can be obtained at discrete-time instants, by sampling from a noisy sensor and estimating $\mf{p}$ with a Kalman filter. In this case, estimates $\{\hat{\mf{p}}_i\}_{i=1}^\mfs{\mfs{N}}$ are correlated such that a sub-optimal information fusion strategy might be adopted. A popular choice is the CIM \cite{Niehsen2002}, yielding an optimization program within the class of problems described by~\eqref{eq:original}. Our method is used to compute the global CIM estimate across the network by the sensor fusion rule
\begin{equation}
\label{eq:sf}
\mkern-18mu\hat{\mf{p}} = \mf{P}(\bm{\lambda})\sum_{i=1}^\mfs{N}\lambda_i\mf{P}_i^{-1}\hat{\mf{p}}_i = \frac{\mf{P}(\bm{\lambda})}{\mfs{N}}\sum_{i=1}^\mfs{N}(\mfs{N}\lambda_i)\mf{P}_i^{-1}\hat{\mf{p}}_i
\end{equation}
Note that $\mfs{N}\lambda_i$ and $\mf{P}(\bm{\lambda})$ can be estimated at agent $i$ using $x_i(t)^2$ and $\hat{\mf{Q}}_i(t)^{-1}$ respectively for $t\geq T_\varepsilon$. Hence, the average in \eqref{eq:sf} for $\{\mfs{N}\lambda_i\mfs{P}_i^{-1}\hat{\mf{p}}_i\}_{i=1}^\mfs{N}$, equivalently $\{x_i^2\mfs{P}_i^{-1}\hat{\mf{p}}_i\}_{i=1}^\mfs{N}$ can be computed by average consensus techniques \cite{Aldana2019,edcho}.}

\RevAll{ The sensor fusion rule in \eqref{eq:sf} has shown to be effective in many applications as described in \cite{Julier2017General,Niehsen2002}, where a sequential solution method is used in \cite{Deng2012sequential} to reduce the computational burden of using all sensors at once at the expense of reduced estimation quality. However, these solutions mostly rely on a centralized node which gathers the information across the network and computes \eqref{eq:sf}. Our proposal is used to compute \eqref{eq:sf} as well, inheriting the same accuracy as \cite{Julier2017General,Niehsen2002,Deng2012sequential} for this task. Nonetheless, we are able to obtain $\hat{\mf{p}}, \mf{P}(\bm{\lambda})$ in a distributed fashion. Other distributed approaches such as \cite{Hu2011Diffusion, Wang2017convergence, He2018Consistent,Sebastian2021CDC} use local versions of the CIM rule in \eqref{eq:sf} where only neighboring sensors are used at each node. In contrast, we are able to include the information of all sensors across the network, which always increases the accuracy when compared to the case in which a subset of the sensors are used \cite{Julier2017General}.}

\section{Convergence analysis}
\label{sec:proof}
In this section, we provide a proof for Theorem \ref{th:main}. Before that, we prove some auxiliary results. \RevOne{Systems \eqref{eq:consensus1} and $\dot{x}_i(t)=u_i(t)$ are discontinuous due to the use of the $\text{sign}(\bullet)$ function and the discontinuous nature of \eqref{eq:ideal} at the boundary of $\mathcal{C}$. Hence, these systems are better understood in the sense of Filippov \cite{Filippov1988}, devised to study discontinuous dynamics.} 

We start by analyzing the consensus protocols in \eqref{eq:consensus1}. Note that all these protocols have the structure:
\begin{equation}
\label{eq:consensus}
\dot{{\upupsilon}}_i(t) =
\kappa_z\sum_{\RevTwo{j}\in\mathcal{N}_i} 
\phi\left(\hat{z}_j(t)-\hat{z}_i(t) ;\zeta,q \right), \ \ \hat{z}_i(t) = z_i(t)-\upupsilon_i(t)
\end{equation}
for suitable input $z_i(t)$, either $x_i(t)^2$ or the elements of $x_i(t)^2\mf{P}_i^{-1}$. \RevOne{Now, we show convergence of \eqref{eq:consensus} before $t=T_c$.}
\begin{proposition}\cite[Special case of Theorem 6]{Aldana2019}
\label{prop:robustConsensus}
Let $\mathcal{G}$ be a connected undirected graph with $\mfs{N}$ nodes, $\ell$ edges, algebraic connectivity $\lambda_{\mathcal{G}}$, and consider 
\begin{equation}
\label{eq:robustCons}
\dot{e}_i(t) = d_i(t)-\kappa_z\sum_{\RevTwo{j}\in\mathcal{N}_i}\phi(e_j(t)-e_i(t);\zeta,q)
\end{equation}
where $|d_i(t)|\leq L', \forall t\geq 0$. Moreover, given $T_c>0$ let $\kappa_z\geq \ell\pi/(q\lambda_\mathcal{G}T_c), \zeta\geq L'/(\kappa_z\sqrt{\lambda_{\mathcal{G}}})$. Then, $e_i(t)\equiv (1/\mfs{N})\sum_{i=1}^\mfs{N}e_i(0), \forall t\geq T_c$.
\end{proposition}
\begin{lemma}
\label{le:consensus}
Let $\mathcal{G}$ be a connected undirected graph with $\mfs{N}$ nodes, $\ell$ edges, algebraic connectivity $\lambda_{\mathcal{G}}$ and $\sum_{i=1}^\mfs{N}\upupsilon_i(0)=0$. Moreover, let $\bar{z}(t) = (1/\mfs{N})\sum_{i=1}^\mfs{N}z_i(t)$ and $|\dot{z}_i(t)|\leq L, \forall t\geq 0$. Then, \eqref{eq:consensus} satisfies that $\hat{z}_i(t)\equiv \bar{z}(t), \forall t\geq T_c$ and $\forall i,$ provided that $\kappa_z\geq \ell\pi/(q\lambda_\mathcal{G}T_c), \zeta\geq 2L/(\kappa_z\sqrt{\lambda_\mathcal{G}})$.
\end{lemma}
\begin{proof}
Let $e_i(t) = \upupsilon_i(t)-(z_i(t)-\bar{z}(t))$ such that:
$\dot{e}_i(t) =
(\dot{\bar{z}}(t) - \dot{z}_i(t))-\kappa_z\sum_{\RevTwo{j}\in\mathcal{N}_i} 
\phi\left(e_j(t)-e_i(t) ;\zeta,q \right)$
equivalent to \eqref{eq:robustCons} with $d_i(t) = \dot{\bar{z}}(t) - \dot{z}_i(t)$. By assumption, $|d_i(t)|\leq 2L$. Thus, Proposition \ref{prop:robustConsensus} is used with $L'=2L$ to conclude that $e_i(t)\equiv(1/\mfs{N})\sum_{i=1}^\mfs{N}e_i(0) = 0, \forall t\geq T_c$. Hence, $\upupsilon_i(t) = z_i(t) - \bar{z}(t)$ and $\hat{z}_i(t) = \bar{z}(t)$ for $t\geq T_c$.
\end{proof}
Now, we study the ideal PGF trajectories for $\mf{x}(t)$, dictated by \eqref{eq:ideal}. \RevOne{The following result shows that $\mf{x}(t)$ converge to the feasible region $\mathcal{C}$, a requirement for a solution of \eqref{eq:relax}.}
\begin{lemma}
\label{le:sliding}
Let \eqref{eq:ideal} and $\underline{b}>0$. Then, $\forall \mf{x}(T_0)\notin\mathcal{C}$, $x_i(T_0)\geq \underline{b}$ there exist $T_\varepsilon>0$ such that $\mf{x}(t)\in\mathcal{C}, \forall t\geq T_0+T_\varepsilon$.
\end{lemma}
\begin{proof} 
We split the proof in two cases: (a) $(1-\varepsilon/2)-s(\mf{x}(T_0))>0$: in this case, note that $\dot{x}_i(t) = \kappa_{\mathcal{C}}x_i(t)$ for $t\in[T_0,T_0+T')$ with $T' = \inf\{t>T_0 : \mf{x}(t)\in\mathcal{C}\}$. Thus, $x_i(t)$ is increasing in such interval and $\|\mf{x}(t)\|\geq \sqrt{\mfs{N}\min_{i\in\mathcal{I}}x_i(t)^2}\geq \sqrt{\mfs{N}}\underline{b}, \quad \forall t\in[T_0,T_0+T').$ Now, consider a Lyapunov function candidate $V_1(\mf{x}(t)) = (1-\varepsilon)-s(\mf{x}(t))$, whose time derivative is $\dot{V}_1(\mf{x}(t)) \kern -0.1cm=\kern -0.1cm -(2/\mfs{N})\mf{x}(t)^\top\dot{\mf{x}}(t) \kern -0.1cm= \kern -0.1cm-(2/\mfs{N})\kappa_\mathcal{C}\|\mf{x}(t)\|^2\kern -0.1cm\leq\kern -0.1cm-2\kappa_{\mathcal{C}}\underline{b}^2.$ Thus, integrating from $T_0$ to $t$, $V_1(\mf{x}(t))\leq V_1(\mf{x}(T_0))-2\kappa_{\mathcal{C}}\underline{b}^2(t-T_0)$ which satisfies $V_1(\mf{x}(T_0+T_\varepsilon))=0$ with $T'\leq T_\varepsilon := V_1(\mf{x}(T_0))/(2\kappa_{\mathcal{C}}\underline{b}^2)$. Hence, $s(\mf{x}(T_0+T_\varepsilon))=1-\varepsilon$ and thus $\mf{x}(T_0+T_\varepsilon)\in\mathcal{C}$.  (b) $(1-\varepsilon/2)-s(\mf{x}(T_0))<0$: note that $\dot{x}_i(t)=-\kappa_\mathcal{C}x_i(t)$ implies $x_i(t)$ is decreasing for $t\in[T_0,T_0+T_\varepsilon)$ and thus, $\|\mf{x}(t)\|=\sqrt{\mfs{N}{s}(\mf{x}(t))}\geq \sqrt{\mfs{N}}$ since $s(\mf{x}(t))\geq 1$. The rest of the argument for this case follows as before using $V_2(\mf{x}) = s(\mf{x})-1$. Thus, $\mf{x}(t)$ reaches $\mathcal{C}$ in finite time. In addition, $\mf{x}(t)$ remain inside $\mathcal{C}$ for $t\geq T_0+T_\varepsilon$ since $V_1(\mf{x}(t))$ (resp. $V_2(\mf{x}(t))$) cannot increase at the inner (resp. outer) boundary of $\mathcal{C}$.
\end{proof}
\RevOne{Now, we recall the definition of  the projected gradient with respect to a manifold and then apply the definition to the problem of interest in \eqref{eq:relax}.}
\begin{definition}\cite{projGrad}
\label{def:proj}
Let $F:\mathbb{R}^n\to\mathbb{R}$ be differentiable and $\mathcal{C}=\{\mf{x}\in\mathbb{R}^n:\gamma_i(\mf{x})\leq 0, i\in\{1,\dots,m\}\}$ for $m$ differentiable constraint functions $\gamma_i:\mathbb{R}^n\to\mathbb{R}$. Then, the projected gradient of $F(\mf{x})$ with respect to $\mathcal{C}$ is given by:
\begin{equation}
\label{eq:projDef}
\begin{aligned}
\proj_\mathcal{C}(\mf{x}, -\nabla \text{$F$}(\mf{x})) &= \argmin_{\mf{w}\in \mathbb{R}^n}\|\mf{w}-(-\nabla F(\mf{x}))\|^2 \\
\text{\normalfont such that} &\ \  \mf{w}^\top\nabla \gamma_i(\mf{x})\leq 0, \forall i\in{\mathcal{J}(\mf{x})}
\end{aligned}
\end{equation}
where $\mathcal{J}(\mf{x})=\{i\in\{1,\dots,m\}:\gamma_i(\mf{x})=0\}$.
\end{definition}

\begin{lemma}
\label{prop:proj}
Let $\mathcal{C}$ with $\gamma_1(\mf{x}) = \|\mf{x}\|^2/\mfs{N} - 1$,  $\gamma_2(\mf{x}) = (1-\varepsilon) - \|\mf{x}\|^2/\mfs{N}$ in Definition \ref{def:proj} with $m=2$. Then, for any differentiable $F(\bullet)$ and any $\mf{x}\in\mathcal{C}$, 
\begin{gather}
\label{eq:projGrad}
\begin{aligned}
&\proj_\mathcal{C}(\mf{x}, -\nabla \text{$F$} (\mf{x})) = \\
&\begin{cases}
-\nabla F(\mf{x}) + \left(\frac{\mf{x}^\top\nabla F(\mf{x})}{\mfs{N}^2}\right)\mf{x} & \gamma_1(\mf{x})=0, \mf{x}^\top\nabla F(\mf{x})\leq 0 \\
-\nabla F(\mf{x}) - \left(\frac{\mf{x}^\top\nabla F(\mf{x})}{(\mfs{N}(1-\varepsilon))^2}\right)\mf{x} & \gamma_2(\mf{x})=0, \mf{x}^\top\nabla F(\mf{x})\geq 0 \\
\mathrlap{-\nabla F(\mf{x})} & \text{\normalfont otherwise}
\\
\end{cases}
\end{aligned}
\raisetag{15pt}
\end{gather}
\end{lemma}
\begin{proof}
First, $\nabla \gamma_1(\mf{x}) = 2\mf{x}/\mfs{N}, \nabla \gamma_2(\mf{x}) = -2\mf{x}/\mfs{N}$. Consider the following cases: (a) $\gamma_1(\mf{x})<0, \gamma_2(\mf{x})<0$: since \eqref{eq:projDef} is unconstrained, $\mf{w}=-\nabla F(\mf{x})$. (b) for $\gamma_1(\mf{x})=0, \mf{x}^\top\nabla F(\mf{x})\geq 0$ or $\gamma_2(\mf{x})=0, \mf{x}^\top\nabla F(\mf{x})\leq 0$: $\mf{w}=-\nabla F(\mf{x})$ satisfies the constraint $\mf{w}^\top\nabla\gamma_1(\mf{x})\leq 0$ and $\mf{w}^\top\nabla\gamma_2(\mf{x})\leq 0$ in \eqref{eq:projDef} respectively. (c) $\gamma_1(\mf{x})=0, \mf{x}^\top\nabla F(\mf{x})\leq 0$: let $h(\mf{w})=\|\mf{w}+\nabla F(\mf{x})\|^2$ and  $\mf{w}^* = -\nabla F(\mf{x}) + \left(\frac{\mf{x}^\top\nabla F(\mf{x})}{\mfs{N}^2}\right)\mf{x}$. Note that $(\mf{w}^*)^\top \nabla\gamma_1(\mf{x})=0$ making $\mf{w}^*$ is feasible for \eqref{eq:projDef} and $\nabla_{\mf{w}}h(\mf{w}) = 2(\nabla F(\mf{x})+\mf{w}) = \lambda\mf{x}/\mfs{N}=(\lambda/2) \nabla_\mf{w}(\mf{w}^\top \nabla\gamma_1(\mf{x}))$ is satisfied with $\mf{w}=\mf{w}^*$ and Lagrange multiplier $\lambda = 2\mf{x}^\top \nabla F(\mf{x})/\mfs{N}$ where $\nabla_{\mf{w}}$ denotes gradient with respect to $\mf{w}$. Hence, $\mf{w}=\mf{w}^*$ satisfies the optimality conditions for \eqref{eq:projDef} and is the unique minimizer \cite[Proposition 3.1.1]{nonlinear_prog}. (d) $\gamma_2(\mf{x})=0, \mf{x}^\top\nabla F(\mf{x})\geq 0 $: the same reasoning of (c) applies.
\end{proof}
\RevOne{In the following, we show that the ideal system \eqref{eq:ideal} flows in the direction of the previously obtained projected gradient.}
\begin{lemma}
\label{le:gradient}
Let \eqref{eq:ideal} with $\mf{x}(T_0)\in\mathcal{C}$. Then, under Definition \ref{def:proj}, $\dot{\mf{x}}(t) = \alpha(\mf{x}(t))\proj_\mathcal{C}(\mf{x}(t), -\nabla \text{$f$} (\mf{Q}(\mf{x}(t)))$ for some positive scalar function $\alpha:\mathbb{R}^n\to\mathbb{R}_{>0}$.
\end{lemma}
\begin{proof}
In this proof, we omit time dependence for conciseness. We compare the right hand side of \eqref{eq:ideal} with \eqref{eq:projGrad} in Lemma \ref{prop:proj} under $F(\mf{x})=f(\mf{Q}(\mf{x}))$: (a) $\gamma_1(\mf{x})<0, \gamma_2(\mf{x})<0$: $\dot{\mf{x}} = -\alpha(\mf{x})\nabla F(\mf{x}) = \alpha(\mf{x})\proj_\mathcal{C}(\mf{x}, -\nabla \text{$F$} (\mf{x}))$ due to \eqref{eq:projGrad} with $\alpha(\mf{x}) = \kappa_\mathcal{C}\mfs{N}$.  (b) $\gamma_1(\mf{x})=0$: we invoke the Filippov interpretation of solutions to write $\dot{\mf{x}}$. Call $\mf{w}_1 = \kappa_\mathcal{C}\mf{x}\text{\normalfont sign}((1-\frac{\varepsilon}{2})-{s}(\mf{x}))$ and $\mf{w}_2 = -\kappa_\mathcal{C}\mfs{N}\nabla f(\mf{Q}(\mf{x}))$. Filippov solutions in this manifold have $\dot{\mf{x}}$ in the convex hull of $\{\mf{w}_1,\mf{w}_2\}$. Concretely, if $\gamma_1(\mf{x})=0, \mf{x}^\top\nabla F(\mf{x})\geq 0$, then $\mf{w}_1, \mf{w}_2$ point inside $\mathcal{C}$, and $s(\mf{x}(t))=0$ is satisfied for an isolated instant $t$, hence ignored by the trajectory \cite{cortes2008}. In the case of $\gamma_1(\mf{x})=0, \mf{x}^\top\nabla F(\mf{x})\leq 0$, $\mf{w}_1, \mf{w}_2$ point in different directions with respect to the inner boundary of $\mathcal{C}$, inducing sliding motion along $\gamma_1(\mf{x})=0$. Therefore, $\dot{\mf{x}}$ is the convex combination of $\mf{w}_1,\mf{w}_2$, lying in the tangent plane of $\gamma_1(\mf{x})=0$ \cite{cortes2008}. Note that $\mf{w}=-\nabla F(\mf{x}) + \left(\frac{\mf{x}^\top\nabla F(\mf{x})}{\mfs{N}^2}\right)\mf{x} = \left(\frac{\mf{x}^\top\nabla F(\mf{x})}{\kappa_\mathcal{C}\mfs{N}^2\text{\normalfont sign}((1-\frac{\varepsilon}{2})-{s}(\mf{x}))}\right)\mf{w}_1 + \left(\frac{1}{\kappa_\mathcal{C}\mfs{N}}\right)\mf{w}_2$ with $\mf{w}^\top \mf{x} = 0$, lying in such tangent plane. Hence, $\dot{\mf{x}}$ is proportional to $\mf{w}$ and $\dot{\mf{x}} = \alpha(\mf{x})\mf{w} = \alpha(\mf{x})\proj_\mathcal{C}(\mf{x}, -\nabla \text{$F$} (\mf{x}))$ for some $\alpha(\mf{x})>0$.  (c) $\gamma_2(\mf{x})=0$: the proof follows from (b).
\end{proof}

\begin{lemma}\label{le:bounded}
Consider \eqref{eq:ideal} under Assumption~\ref{ass:boundness} and $f(\bullet)$ in Table \ref{tab:options} labeled by $\mu\in\{0,1,2\}$. For any $\mf{x}(T_0)\notin\mathcal{C}$ it follows that $\|\dot{\mf{x}}(t)\|\leq h(\mfs{N}), \forall t\geq T_0$ with $h(\bullet)$ as in Theorem \ref{th:main}.
\end{lemma}
First, we show that $|x_i(t)|\leq \overline{b}, \forall t\geq T_0$ by contradiction. Assume there exists $T_2\geq T_0$ such that $|x_i(T_2)|>\overline{b}$. By continuity of the solution, there must have existed $T_1=\sup\{t\in[T_0,T_2]: |x_i(t)|=\overline{b}\}$ such that $a_i(t) = |x_i(t)|$ is increasing for $t\in[T_1,T_2]$. However, in this interval, $\mf{x}(t)\notin\mathcal{C}$ since $s(\mf{x}(t))=\|\mf{x}(t)\|^2/\mfs{N}\geq \min_{i\in\mathcal{I}}|x_i(t)|^2/\mfs{N}\geq \overline{b}^2>1$. Therefore, due to \eqref{eq:ideal}, $\dot{a}_i(t) = -\kappa_\mathcal{C}a_i(t)$ which means that $a_i(t)$ is decreasing in $t\in[T_1,T_2]$ leading to a contradiction. 

Now, we show 
$
\tr(\mf{Q}(\mf{x}(t))^{-1})\leq \mfs{N}(\underline{\sigma}\min\{\underline{b}^2,1-\varepsilon\})^{-1}, \forall t\geq T_0.
$
From the proof of Lemma \ref{le:sliding} it follows that $s(\mf{x}(t))=\|x_i(t)\|^2/\mfs{N}\geq \underline{b}^2, \forall t\in[T_0,T_0+T_\varepsilon]$ and ${s}(\mf{x}(t))\geq 1-\varepsilon, \forall t\geq T_0+T_\varepsilon$. Henceforth, $\mf{Q}(\mf{x}(t))\succeq (1/\mfs{N})\sum_{i=1}^{\mfs{N}}x_i(t)^2\underline{\sigma}\mf{I}\succeq \underline{\sigma}\min\{\underline{b}^2,1-\varepsilon\}\mf{I}$ for $t\geq T_0$. Then, $\mf{Q}(\mf{x}(t))^{-1}\preceq (\underline{\sigma}\min\{\underline{b}^2,1-\varepsilon\})^{-1}\mf{I}$ from which $\tr(\mf{Q}(\mf{x}(t))^{-1})\leq \mfs{N}(\underline{\sigma}\min\{\underline{b}^2,1-\varepsilon\})^{-1}$ is obtained. From the previous bounds, and the gradient in Table \ref{tab:options}:
$$\begin{aligned} \|\nabla f(\mf{Q}(\mf{x}(t)))\|&\leq \max_{i\in\mathcal{I}}2|x_i|\tr(\mf{P}_i^{-1})\tr(\mf{Q}(\mf{x}(t))^{-1})^{\mu} \\&\leq 2\overline{b}\overline{\sigma}\mfs{N}^{\mu}(\underline{\sigma}\min\{\underline{b}^2,1-\varepsilon\})^{-\mu},\end{aligned}$$ and $$\begin{aligned}\|\dot{\mf{x}}(t)\|&\leq \kappa_\mathcal{C}\max\{\|\mf{x}(t)\|, \mfs{N}\|\nabla f(\mf{Q}(\mf{x}(t)))\|\} \\& \leq \kappa_\mathcal{C}\max\{\sqrt{\mfs{N}}\overline{b}, 2\overline{b}\overline{\sigma}\mfs{N}^{\mu+1}(\underline{\sigma}\min\{\underline{b}^2,1-\varepsilon\})^{-\mu}\}\qed\end{aligned}$$

\RevOne{Recalling from Lemma \ref{le:sliding} that \eqref{eq:ideal} flows in the direction of the projected gradient, the following result characterizes convergence of such flow towards the optimum of \eqref{eq:relax}.}

\begin{proposition}
\label{prop:projFlow}
Let $F(\bullet),\gamma_i(\bullet),\mathcal{C}$ defined as in Definition \ref{def:proj} and $\alpha:\mathbb{R}^n\to\mathbb{R}_{>0}$ be a positive scalar function such that $\dot{\mf{x}}(t)=\alpha(\mf{x}(t))\proj_\mathcal{C}(\mf{x}(t),- \nabla \text{$f$}(\mf{x}(t)))$ has unique forward solution for $t\geq T_0$ and any $\mf{x}(T_0)\in\mathcal{C}$. Then, $\mf{x}(t)$ converges asymptotically towards $\argmin_{\mf{x}\in\mathcal{C}}F(\mf{x})$.
\end{proposition}
\begin{proof}
The result follows using the same reasoning in  \cite[Theorem 3]{projGrad}, where a proof was provided for $\alpha(\mf{x})=1$.
\end{proof}


\noindent\textbf{Proof of Theorem \ref{th:main}}: We start by analyzing \eqref{eq:consensus1} for $t\in[0,T_c]$, in which $x_i(t)^2\mf{P}_i^{-1}$ and $x_i(t)^2$ remain constant because $u_i(t)=0$. The result from Lemma \ref{le:consensus} is valid obtaining $\hat{s}_i(t)\equiv s(\mf{x}(t)), \hat{\mf{Q}}_i(t) \equiv \mf{Q}(\mf{x}(t))$ at $t=T_c$. Assume that the condition $\hat{s}_i(t)\equiv s(\mf{x}(t))$ is  maintained for some open interval $\mathcal{T}\subset[T_c,\infty)$. Then, $\mf{x}(t)$ is dictated by \eqref{eq:ideal} for any $t\in\mathcal{T}$ and thus $z_i(t)=x_i(t)^2$ comply $|\dot{z}_i(t)|\leq 2|x_i(t)|\|\dot{\mf{x}}(t)\|\leq 2\overline{b}h(\mfs{N})$ due to Lemma \ref{le:bounded}. Hence, the conditions of Lemma \ref{le:consensus} are fulfilled with $L=2\overline{b}h(\mfs{N})$ due to the choice of $\kappa_s$ in Theorem \ref{th:main} for such interval $\mathcal{T}$. 

Now, we show that $\hat{s}_i(t)\equiv s(\mf{x}(t))$ is maintained $\forall t\in[T_c, \infty)$. We verify this by contradiction. Assume that there exists $T_2=\inf\{t\geq T_c: \|\dot{\mf{x}}(t)\|>2\overline{b}h(\mfs{N})\}$. Then, there must have existed some $T_1=\sup\{t<T_2 : \hat{s}_i(t)\neq s(\mf{x}(t)) \text{ or } \hat{\mf{Q}}_i(t)\neq \mf{Q}(\mf{x}(t))\}$ since the system could not have been the ideal one in \eqref{eq:ideal}. Hence, the system is not synchronized for some time $t\in[T_1,T_2]$ for which $\|\mf{x}(t)\|\leq 2\overline{b}h(\mfs{N})$, which is impossible since this means that the inputs for the consensus protocols have bounded derivative, and Lemma \ref{le:consensus} ensure synchronization. Hence, $\hat{s}_i(t)\equiv s(\mf{x}(t)), \forall t\geq T_c$. The same reasoning applies to the synchronization condition $\hat{\mf{Q}}_i(t)\equiv \mf{Q}(\mf{x}(t))$.

Now, use Lemma \ref{le:sliding} with $T_0=T_c$ to obtain $\mf{x}(t)\in\mathcal{C}, \forall t\geq T_c+T_\varepsilon$. We can verify $\mathcal{E}(\mf{Q}(\mf{x}(t))^{-1})\supset \check{\mathcal{E}}, \forall t\geq T_c+T_\varepsilon$ by taking any $\mf{y}\in\check{\mathcal{E}}$ which satisfy $\mf{y}^\top\mf{Q}(\mf{x}(t))\mf{y} = \sum_{i=1}^\mfs{N}(x_i(t)^2/\mfs{N})\mf{y}^\top\mf{P}_i^{-1}\mf{y}\leq \sum_{i=1}^\mfs{N}(x_i(t)^2/\mfs{N}) = s(\mf{x}(t))\leq 1$. Hence, $\mf{y}\in\mathcal{E}(\mf{Q}(\mf{x}(t))^{-1})$. Use Lemma \ref{le:gradient} with $T_0=T_c+T_\varepsilon$ to write \eqref{eq:ideal} as $\dot{\mf{x}}(t) = \alpha(\mf{x}(t))\proj_\mathcal{C}(\mf{x}(t), -\nabla \text{$f$} (\mf{Q}(\mf{x}(t)))$, implying convergence towards the optimal weights due to Proposition \ref{prop:projFlow}. \qed  

\section{Numerical experiments}\label{sec:simulated_experiments}

In this section, we simulate the proposal to illustrate its properties \RevOne{in the context of a sensor fusion problem. We assume that each agent reads a noisy measurement from a sensor with covariance matrix $\{\mf{P}_i\}_{i=1}^\mfs{N}$. For the sake of generality, we set $\mf{P}_i = \mf{M}_i^\top\mf{M}_i$ where $\mf{M}_i$ was drawn with uniformly distributed components. The purpose of each agent $i$ is to compute $\mfs{N}\lambda_i, \mf{P}(\bm{\lambda})$ from $x_i(t)^2, \hat{\mf{Q}}_i(t)^{-1}$ so that they can be used in the sensor fusion rule in \eqref{eq:sf}. We choose $f(\mf{Q}(\mf{x})) = \tr(\mf{Q}(\mf{x})^{-1})$ as a performance index.}

All the differential equations, namely, \eqref{eq:consensus1} and $\dot{x}_i(t)=u_i(t)$ under \eqref{eq:controller}, where simulated using the forward Euler method with time step $\Delta t=10^{-4}$. In the sake of interpretability, we set $n=2$ and a graph $\mathcal{G}$ of $\mfs{N}=6$ agents shown in Figure \ref{fig:graph}. In addition, we set $\varepsilon=0.05$ as well as parameters $T_c=1, q=1/2, \underline{b}=0.1, \overline{b}=1.1, \kappa_\mathcal{C}=0.1, \kappa_s=\kappa_\mf{Q}=10,\zeta_{\mf{Q}}=\zeta_s=1, \zeta_g=4$. 
\begin{figure}
\vspace{0.5em}
\centering
\includegraphics[width=0.3\textwidth]{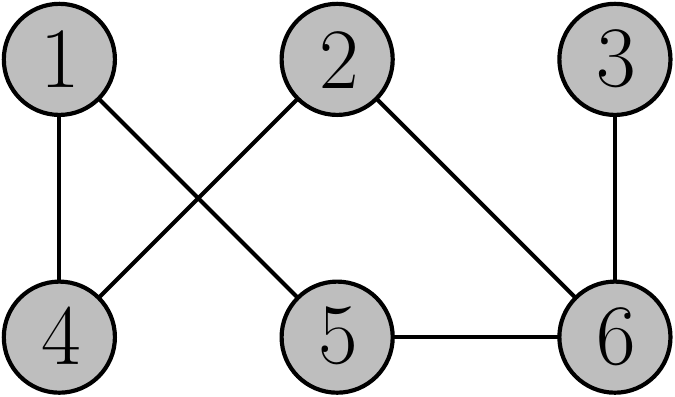}
\caption{Undirected graph $\mathcal{G}$ used in Section \ref{sec:simulated_experiments}.}
\label{fig:graph}
\end{figure}
Figure \ref{fig:surface} shows the local estimates $\hat{s}_i(t)$ which reach consensus towards the global value of $s(\mf{x}(t))$ at $t\approx 0.05$, occurring before $t=T_c=1$, which is a conservative prescribed convergence time bound. After $t=T_c$,  $s(\mf{x}(t))$ increases until it reaches $s(\mf{x}(t))=1-\varepsilon$, which means that $\mf{x}(t)\in\mathcal{C}$ from $t\geq T_c+T_\varepsilon$ with $T_\varepsilon\approx 0.57$. Figure \ref{fig:convergence} shows that, once $\mf{x}(t)\in\mathcal{C}$, then, the trajectories of $\lambda_i(t)=x_i(t)^2/\mfs{N}$ converge towards the global minimizer of \eqref{eq:relax} which corresponds to $\lambda_1^*=\lambda_2^*=\lambda_3^*=\lambda_4^*=0, \lambda_5^*=0.6582, \lambda_6^*=0.3418$ in this example, up to some error less than $\varepsilon=0.05$. Finally, Figure \ref{fig:ellipses} show the ellipses $\{\mathcal{E}(\mf{P}_i)\}_{i=1}^\mfs{N}$ as well as $\mathcal{E}(\mf{Q}(\mf{x}(t))^{-1})$ with $t=T_c+T_\varepsilon$ and $t=30$, which shows that $\mathcal{E}(\mf{Q}(\mf{x}(t))^{-1})\supset \check{\mathcal{E}}$ in both cases, but the latter case leads to a tighter outer L-J ellipse.

\RevTwo{Moreover, we evaluate the time it takes for the agents to reach consensus $T_{\mfs{cons}}$ and $T_\varepsilon$ by repeating the same experiment as before with a circular graph $\mathcal{G}$ of $\mfs{N}=10$ and $20$ agents and with $n=2$ and $4$. We set $\kappa_s,\kappa_\mfs{Q}, \zeta_s,\kappa_\mfs{Q}$ the same in all cases, adjusted accordingly to account for $\ell,\lambda_{\mathcal{G}}$ in both networks as required in Theorem \ref{th:main}. The results are summarized in Table \ref{tab:experiment} where it is observed that $T_{\mfs{cons}}$ increases with $\mfs{N}$ as expected from \cite{Aldana2019}. On the other hand,  $T_{\mfs{cons}}$ is similar as before when $n$ is increased. The reason is that \eqref{eq:consensus1} is executed in parallel for each component of $\hat{\mf{Q}}_i(t)$ such that there is no influence of $n$ in the protocol convergence itself. Note that $T_\varepsilon$ is similar in all cases since once consensus has been reached, the system is equivalent to \eqref{eq:ideal} which does not depend on $\mfs{N},n$ for $\mf{x}(t)\notin\mathcal{C}$.}

\begin{figure}
\centering
\includegraphics[width=0.45\textwidth]{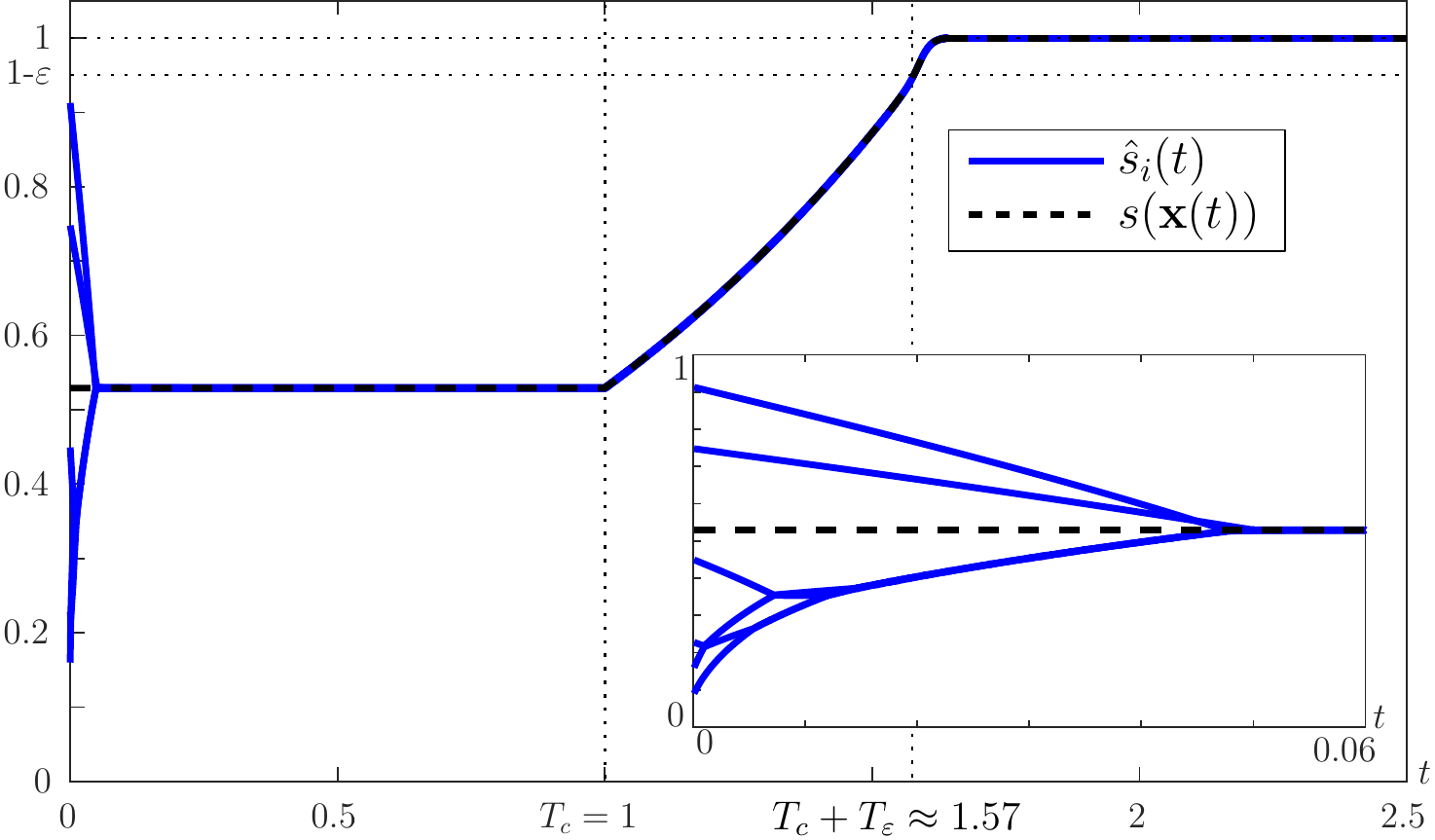}
\caption{Convergence of the estimates $\hat{s}_i(t)$ which reach consensus towards $s(\mf{x}(t))$ at $t\approx 0.05$, with zoom in $t\in[0,0.06]$. Moreover,  the evolution of the surface $s(\mf{x}(t))$ towards achieving $\mf{x}(t)\in\mathcal{C}$ is shown.}
\label{fig:surface}
\end{figure}

\begin{figure}
\centering
\vspace{0.5em}
\includegraphics[width=0.45\textwidth]{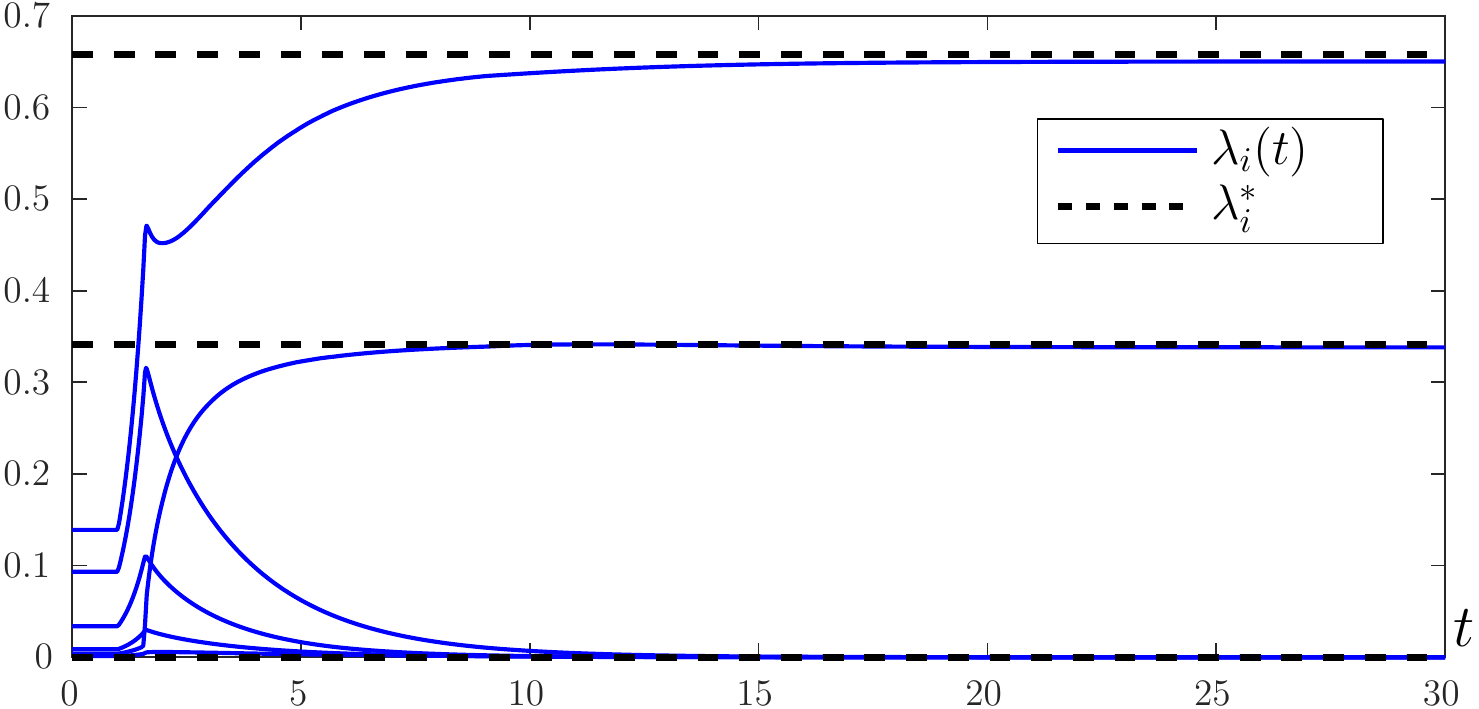}
\caption{Evolution of the trajectories $\lambda_i(t)=x_i(t)^2/\mfs{N}$ which converge towards the global optimal values of \eqref{eq:original} asymptotically, up to error less than $\varepsilon=0.05$.}
\label{fig:convergence}
\end{figure}

\begin{figure}
\centering
\includegraphics[width=0.4\textwidth]{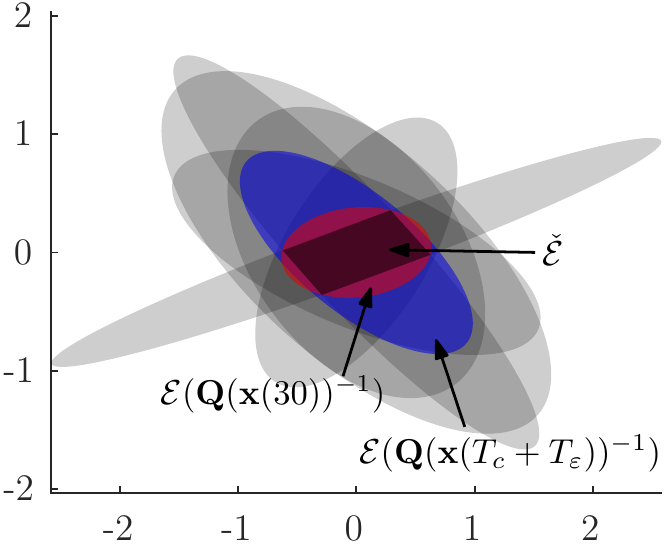}
\caption{Ellipses $\{\mathcal{E}(\mf{P}_i)\}_{i=1}^\mfs{N}$ (shown in grey) as well as $\mathcal{E}(\mf{Q}(\mf{x}(t))^{-1})$ with $t=T_c+T_\varepsilon$ (blue) and $t=30$ (red).}
\label{fig:ellipses}
\end{figure}

\begin{table}[!ht]
\centering
\renewcommand{\arraystretch}{1.5}
\begin{tabular}{llllll}
\hline
\multicolumn{1}{||l|}{$\mfs{N}$}                 & \multicolumn{1}{l|}{$\ell$}                 & \multicolumn{1}{l|}{$\lambda_{\mathcal{G}}$}                & \multicolumn{1}{l|}{$n$} & \multicolumn{1}{l|}{$T_{\mfs{cons}}\ \mfs{[s]}$} & \multicolumn{1}{l||}{$T_\varepsilon\ \mfs{[s]}$} 
\\ 
\hline
\multicolumn{1}{||l|}{\multirow{2}{*}{10}} & \multicolumn{1}{l|}{\multirow{2}{*}{9}} & \multicolumn{1}{l|}{\multirow{2}{*}{0.3819}} & \multicolumn{1}{l|}{2}   & \multicolumn{1}{l|}{0.0371} & \multicolumn{1}{l||}{0.169} \\ \cline{4-6} 
\multicolumn{1}{||l|}{}                  & \multicolumn{1}{l|}{}                  & \multicolumn{1}{l|}{}                  & \multicolumn{1}{l|}{4}   & \multicolumn{1}{l|}{0.0368} & \multicolumn{1}{l||}{0.172} \\ \hline
\multicolumn{1}{||l|}{\multirow{2}{*}{20}} & \multicolumn{1}{l|}{\multirow{2}{*}{19}} & \multicolumn{1}{l|}{\multirow{2}{*}{0.0978}} & \multicolumn{1}{l|}{2}   & \multicolumn{1}{l|}{0.085} & \multicolumn{1}{l||}{0.170} \\ \cline{4-6} 
\multicolumn{1}{||l|}{}                  & \multicolumn{1}{l|}{}                  & \multicolumn{1}{l|}{}                  & \multicolumn{1}{l|}{4}   & \multicolumn{1}{l|}{0.088} & \multicolumn{1}{l||}{0.181} \\ \hline
\end{tabular}
\caption{\RevTwo{Values for the consensus and feasible region reaching times $T_\mfs{cons}, T_\varepsilon$, for a circular graph with $\mfs{N}$ nodes.}}
\label{tab:experiment}
\renewcommand{\arraystretch}{1}
\end{table}

\section{Conclusions}\label{sec:conclusions}

This work has developed a distributed method to compute a class of outer L-J ellipsoids. This is particularly useful for sensor fusion applications, as it allows the computation of the covariance intersection ellipsoid by combining information from the entire network using only local interactions. The algorithm reformulates the centralized problem and utilizes EDC tools to reach consensus and converge to an outer L-J ellipsoid in finite time, and to the global optimum asymptotically. The proposal's advantages are supported by formal convergence analysis and numerical experiments.



\balance

\end{document}